\documentclass[aps,reprint,floatfix]{revtex4-1}
\usepackage{amsmath,amssymb,graphicx,color,mathtools}
\usepackage{calligra,array,bbm}

\newcounter{theoremcounter}
\stepcounter{theoremcounter}

\newtheorem{theorem}{Theorem}
\newtheorem{lemma}{Lemma}

\newtheorem{corollary}{Corollary}

\newenvironment{proof}[1][Proof]{\begin{trivlist}
\item[\hskip \labelsep {\bfseries #1}]}{\end{trivlist}}

\newcommand{\qed}{\hfill $\blacksquare$}

\newcommand{\bs}{\boldsymbol}
\newcommand{\bra}[1]{\left\langle #1\right|}
\newcommand{\ket}[1]{\left|#1\right\rangle}

\newcommand{\expval}[3]{\left\langle #1\middle|#2\middle|#3\right\rangle}
\newcommand{\ketbra}[2]{\ket{#1}\bra{#2}}
\newcommand{\partder}[2]{\frac{\partial #1}{\partial #2}}
\DeclareMathOperator{\Tr}{Tr}

\let\Re\relax
\DeclareMathOperator{\Re}{\text{Re}}
\let\Im\relax
\DeclareMathOperator{\Im}{\text{Im}}

\makeatletter
\newcommand{\vast}{\bBigg@{4}}
\newcommand{\Vast}{\bBigg@{5}}
\makeatother

\makeatletter
\newcommand{\thickhline}{%
    \noalign {\ifnum 0=`}\fi \hrule height 1pt
    \futurelet \reserved@a \@xhline
}
\newcolumntype{;}{@{\hskip\tabcolsep\vrule width 1pt\hskip\tabcolsep}}
\makeatother

\begin{document}

\title{Semiclassical Formulation of Gottesman-Knill and Universal Quantum Computation}
\author{Lucas Kocia}
\affiliation{Department of Physics, Tufts University, Medford, Massachusetts 02155, U.S.A.}
\author{Yifei Huang}
\affiliation{Department of Physics, Tufts University, Medford, Massachusetts 02155, U.S.A.}
\author{Peter Love}
\affiliation{Department of Physics, Tufts University, Medford, Massachusetts 02155, U.S.A.}
\begin{abstract}
  We give a path integral formulation of the time evolution of qudits of odd dimension. This allows us to consider semiclassical evolution of discrete systems in terms of an expansion of the propagator in powers of \(\hbar\). The largest power of \(\hbar\) required to describe the evolution is a traditional measure of classicality. We show that the action of the Clifford operators on stabilizer states can be fully described by a single contribution of a path integral truncated at order \(\hbar^0\) and so are ``classical,'' just like propagation of Gaussians under harmonic Hamiltonians in the continuous case. Such operations have no dependence on phase or quantum interference. Conversely, we show that supplementing the Clifford group with gates necessary for universal quantum computation results in a propagator consisting of a finite number of semiclassical path integral contributions truncated at order \(\hbar^1\), a number that nevertheless scales exponentially with the number of qudits. The same sum in continuous systems has an infinite number of terms at order \(\hbar^1\).
\end{abstract}
\maketitle

\section{Introduction}
\label{sec:intro}

The study of contextuality in quantum information has led to progress in our understanding of the Wigner function for discrete systems. Using Wootters' original derivation of discrete Wigner functions~\cite{Wootters87}, Eisert~\cite{Mari12}, Gross~\cite{Gross06}, and Emerson~\cite{Howard14} have pushed forward a new perspective on the quantum analysis of states and operators in finite Hilbert spaces by considering their quasiprobability representation on discrete phase space. Most notably, the positivity of such representations has been shown to be equivalent to non-contextuality, a notion of classicality~\cite{Howard14,Spekkens08}. Quantum gates and states that exhibit these features are the stabilizer states and Clifford operations used in quantum error correction and stabilizer codes. The non-contextuality of stabilizer states and Clifford operations explains why they are amenable to efficient classical simulation~\cite{Gottesman98,Aaronson04}.

This progress raises the question of how these discrete techniques are connected to prior established methods for simulating quantum mechanics in phase space. A particularly relevant method is trajectory-based semiclassical propagation, which has been widely used in the continuous context. Perhaps, when applied to the discrete case, semiclassical propagators can lend their physical intuition to outstanding problems in quantum information. Conversely, concepts from quantum information may serve to illuminate the comparatively older field of continuous semiclassics.

Quantum information attempts to classify the ``quantumness'' of a system by the presence or absence of various quantum resources. Semiclassical analysis proceeds by successive approximation using \(\hbar\) as a small parameter, where the power of \(\hbar\) required is a measure of ``quantumness''. Can these two views of quantum vs. classical be related? In the current paper, we build a bridge from the continuous semiclassical world to the discrete world and examine the classical-quantum characteristics of discrete quantum gates found in circuit models and their stabilizer formalism.

Stabilizer states are eigenvalue one eigenvectors of a commuting set of operators making up a group which does not contain \(-\mathbb I\). The set of stabilizer states is preserved by elements of the Clifford group, which is the normalizer of the Pauli group, and can be simulated very efficiently. More precisely, by the Gottesman-Knill Theorem, for \(n\) qubits, a quantum circuit of a Clifford gate can be simulated using \(\mathcal{O}(n)\) operations on a classical computer. Measurements require \(\mathcal{O}(n^2)\) operations with \(\mathcal{O}(n^2)\) bits of storage~\cite{Gottesman98,Aaronson04}.

The reason that stabilizer evolution by Clifford gates can be efficiently simulated classically has been explained in various ways. For instance, as already mentioned, stabilizer states have been shown to be non-contextual in qudit~\cite{Howard14} and rebit~\cite{Delfosse15} systems. The obstacle to proving this for qubits is that qubit systems possess state-independent contextuality~\cite{Mermin93}. Of course, we know how to simulate qubit stabilizer states and Clifford operations efficiently by the Gottesmann-Knill theorem~\cite{Aaronson04}. For recent progress relating non-contextuality to classical simulatability for qubits we refer the reader to~\cite{Raussendorf15}. It has also been shown for dimensions greater than two that a state of a discrete system is a stabilizer state if and only if its appropriately defined discrete Wigner function is non-negative~\cite{Gross06}. Therefore, when acted on by positive-definite operators, it can be considered as a proper positive-definite (classical) distribution.

Here, we instead relate the concept of efficient classical simulation to the power of \(\hbar\) that a path integral treatment must be expanded to in order to describe the quantum evolution of interest. It is well known that Gaussian propagation in continuous systems under harmonic Hamiltonians can be described with a single contribution from the path integral truncated at order \(\hbar^0\)~\cite{Heller75}. We show that the corresponding case in discrete systems exists. In the discrete case, stabilizer states take the place of Gaussians and harmonic Hamiltonians that additionally preserve the discrete phase space take the place of the general continuous harmonic Hamiltonians. In the discrete case we will only consider \(d\)-dimensional systems for odd \(d\) since their center representation (or Weyl formalism) is far simpler.

As a consequence, we will show that operations with Clifford gates on stabilizer states can be treated by a path integral independent of the magnitude of \(\hbar\) and are thus fundamentally classical. Such operations have no dependence on phase or quantum interference. This can be viewed as a restatement of the Gottesman-Knill theorem in terms of powers of \(\hbar\). 

We also consider more general propagation for discrete quantum systems. Quantum propagation in continuous systems can be treated by a sum consisting of an infinite number of contributions from the path integral truncated at order \(\hbar^1\). In discrete systems, we show that the corresponding sum consists of a finite number of terms, albeit one that scales exponentially with the number of qudits.

This work also answers a question posed by the recent work of Penney \emph{et al}. that explored a ``sum-over-paths'' expression for Clifford circuits in terms of symplectomorphisms on phase-space and associated generating actions. Penney \emph{et al}. raised the question of how to relate the dynamics of the Wigner representation of (stabilizer) states to the dynamics which are the solutions of the discrete Euler-Lagrange equations for an associated functional~\cite{Penney16}. By relying on the well-established center-chord (or Wigner-Weyl-Moyal) formalism in continuous~\cite{Almeida98} and discrete systems~\cite{Rivas99}, we show how the dynamics of Wigner representations are governed by such solutions related to a ``center generating'' function and that these solutions are harmonic and classical in nature.

We begin by giving an overview of the center-chord representation in continuous systems in Section~\ref{sec:contcenterchord}. Then, Section~\ref{sec:contpathintegral} introduces the expansion of the path integral in powers of \(\hbar\). This leads us to show what ``classical'' simulability of states in the continuous case corresponds to and to what higher order of \(\hbar\) an expansion is necessary to treat any quantum operator. Section~\ref{sec:discrete} then introduces the discrete variable case and defines its corresponding conjugate position and momentum operators. The path integral in discrete systems is then introduced in Section~\ref{sec:discretesemi} and, in Section~\ref{sec:stabilizergroup}, we define the Clifford group and stabilizer states. We prove that stabilizer state propagation within the Clifford group is captured fully up to order \(\hbar^0\) and so is efficiently simulable classically. Section~\ref{sec:discreteuniversalcomp} shows that extending the Clifford group to a universal gate set necessitates an expansion of the semiclassical propagator to a finite sum at order \(\hbar^1\). Finally, we close the paper with some discussion and directions for future work in Section~\ref{sec:conc}.

\section{Center-Chord Representation in Continuous Systems}
\label{sec:contcenterchord}

We define position operators \(\hat q\), \(\hat q \ket{q'} = q' \ket{q'}\), and momentum operators \(\hat p\) as their Fourier transform, \(\hat p = \hat{\mathcal F}^\dagger \hat q \hat{\mathcal F}\),
where
\begin{equation}
  \hat F = h^{\frac{n}{2}} \int^\infty_{-\infty} \mbox d \bs p \int^\infty_{-\infty} \mbox d \bs q \exp \left(\frac{2 \pi i}{\hbar} \bs p \cdot \bs q \right) \ketbra{\bs p}{\bs q}.
  \label{eq:fourier}
\end{equation}

Since \([\hat q, \hat p] = i \hbar\), these operators produce a particularly simple Lie algebra and are the generators of a Lie group. In this Lie group we can define the ``boost'' operator:
\begin{equation}
  \hat Z^{\delta p} \ket{q'} = e^{\frac{i}{\hbar} \hat q \delta p} \ket{q'} = e^{\frac{i}{\hbar} q' \delta p} \ket{q'},
\end{equation}
and the ``shift'' operator:
\begin{equation}
  \hat X^{\delta q} \ket{q'} = e^{-\frac{i}{\hbar} \hat p \delta q} \ket{q'} = \ket{q' + \delta q}.
\end{equation}

Using the canonical commutation relation and \(e^{\hat A + \hat B} = e^{\hat A} e^{\hat B} e^{-\frac{1}{2}[\hat A, \hat B]}\) if \([\hat A, \hat B]\) is a constant, it follows that
\begin{equation}
  \label{eq:contWeylrelation}
  \hat Z \hat X =  e^{\frac{i}{\hbar}} \hat X \hat Z.
\end{equation}
This is known as the Weyl relation and shows that the product of a shift and a boost (a generalized translation) in phase space is only unique up to a phase governed by \(\hbar\).

We proceed to introduce the chord representation of operators and states~\cite{Almeida98}. The generalized phase space translation operator (often called the Weyl operator) is defined as a product of the shift and boost:
\begin{equation}
\hat T(\bs \xi_p, \bs \xi_q) = e^{-\frac{i}{2 \hbar} \bs \xi_p \cdot \bs \xi_q} \hat Z^{ \bs \xi_p} \hat X^{ \bs \xi_q},
\end{equation}
where \(\bs \xi \equiv (\bs \xi_p, \bs \xi_q) \in \mathbb{R}^{2n}\) define the chord phase space. \(\hat T(\bs \xi_p, \bs \xi_q)\) is a translation by the chord \(\bs \xi\) in phase space. This can be seen by examining its effect on position and momentum states:
\begin{equation}
  \hat T(\bs \xi_p, \bs \xi_q) \ket{\bs q} = e^{\frac{i}{\hbar} \left(\bs q + \frac{\bs \xi_q}{2}\right) \cdot \bs \xi_p} \ket{\bs q+\bs \xi_q},
\end{equation}
and
\begin{equation}
  \hat T(\bs \xi_p, \bs \xi_q) \ket{\bs p} = e^{-\frac{i}{\hbar} \left(\bs p + \frac{\bs \xi_p}{2}\right) \cdot \bs \xi_q} \ket{\bs p+\bs \xi_p},
\end{equation}
which are shown in Fig.~\ref{fig:translations}. Changing the order of shifts \(\hat X\) and boosts \(\hat Z\) changes the phase of the translation in phase space by \(\bs \xi\), as given by the Weyl relation above (Eq.~\ref{eq:contWeylrelation}).

An operator \(\hat A\) can be expressed as a linear combination of these translations:
\begin{equation}
  \hat A = \int^\infty_{-\infty} \mbox{d} \bs \xi_p \int^\infty_{-\infty} \mbox{d} \bs \xi_q \, A_\xi(\bs \xi_p, \bs \xi_q) \hat T(\bs \xi_p, \bs \xi_q),
\end{equation}
where the weights are:
\begin{equation}
  A_\xi(\bs \xi_p, \bs \xi_q) = \Tr \left( \hat T(\bs \xi_p, \bs \xi_q)^\dagger \hat A \right).
\end{equation}
These weights give the chord representation of \(\hat A\).

\begin{figure}[ht]
\includegraphics[scale=1.0]{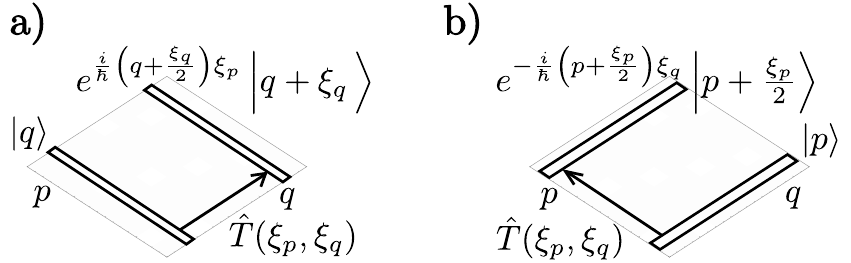}
\caption{Translation of a) a position state and b) a momentum state along the chord \((\xi_p, \xi_q)\) in phase space.}
\label{fig:translations}
\end{figure}

The Weyl function, or center representation, is dual to the chord representation. It is defined in terms of reflections instead of translations. We can define the reflection operator \(\hat R\) as the symplectic Fourier transform of the translation operator:
\begin{eqnarray}
  \label{eq:contreflection}
  &&\hat R(\bs x_p, \bs x_q) = \left(2 \pi \hbar\right)^{-n} \int^\infty_{-\infty} \mbox{d} \bs \xi e^{\frac{i}{\hbar} {\bs \xi}^T \bs{\mathcal J} {\bs x} } \hat T(\bs \xi)
\end{eqnarray}
where \(\bs x \equiv \left(\bs x_p, \bs x_q\right) \in \mathbb{R}^{2n}\) are a continuous set of Weyl phase space points or centers and \(\bs{\mathcal J}\) is the symplectic matrix
\begin{equation}
\bs{\mathcal J} =  \left( \begin{array}{cc} 0 & -\mathbb{I}_{n}\\ \mathbb{I}_{n} & 0 \end{array}\right),
\end{equation}
for \(\mathbb{I}_n\) the \(n\)-dimensional identity. The association of this operator with reflection can be seen by examining its effect on position and momentum states:
\begin{equation}
  \label{eq:contreflectpos}
  \hat R(\bs x_p, \bs x_q) \ket{\bs q} = e^{\frac{i}{\hbar} 2 (\bs x_q - \bs q) \cdot \bs x_p} \ket{2\bs x_q-\bs q},
\end{equation}
and
\begin{equation}
  \label{eq:contreflectmom}
  \hat R(\bs x_p, \bs x_q) \ket{\bs p} = e^{-\frac{i}{\hbar} 2 (\bs x_p - \bs p) \cdot \bs x_q} \ket{2\bs x_p-\bs p},
\end{equation}
which are sketched in Fig.~\ref{fig:reflections}. It is thus evident that \(\hat R(\bs x_p, \bs x_q)\) reflects the phase space around \(\bs x\).  Note that while we refer to reflections in the symplectic sense here, and in the rest of the paper; Eqs.~\ref{eq:contreflectpos} and~\ref{eq:contreflectmom} show that they are in fact ``an inversion around \(\bs x\)'' in every two-plane of conjugate \({x_p}_i\) and \({x_q}_i\). However, we will keep to the established nomenclature~\cite{Almeida98}.

\begin{figure}[ht]
\includegraphics[scale=1.0]{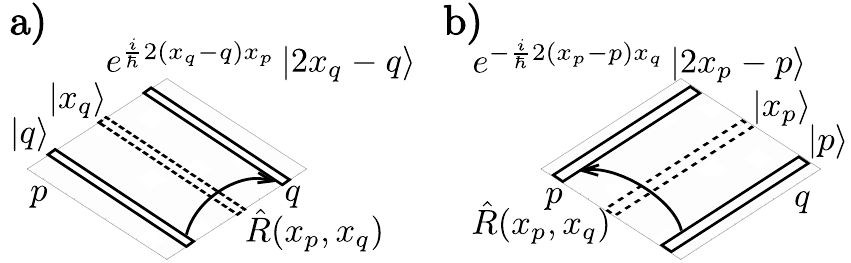}
\caption{Reflection of a) a position state and b) a momentum state across the center \((x_p, x_q)\) in phase space.}
\label{fig:reflections}
\end{figure}

An operator \(\hat A\) can now be expressed as a linear combination of reflections:
\begin{equation}
  \hat A = \left(2 \pi \hbar\right)^{-n} \int^\infty_{-\infty} \mbox{d} \bs x_p \int^\infty_{-\infty} \mbox{d} \bs x_q \, A_x(\bs x_p, \bs x_q) \hat R( \bs x_p, \bs x_q ),
  \label{eq:contsupofreflections}
\end{equation}
where
\begin{equation}
  A_x(\bs x_p, \bs x_q) = \Tr \left( \hat R(\bs x_p, \bs x_q)^\dagger \hat A \right),
  \label{eq:contweylfunction}
\end{equation}
and is called the center representation of \(\hat A\).

This representation is of particular interest to us because we can rewrite the components \(A_x\) for unitary transformations \(\hat A\) as:
\begin{equation}
A_x(\bs x_p, \bs x_q) = e^{\frac{i}{\hbar} S(\bs x_p, \bs x_q)},
\end{equation}
where \(S(\bs x_p, \bs x_q)\) is equivalent to the action the transformation \(A_x\) produces in Weyl phase space in terms of reflections around centers \(\bs x\)~\cite{Almeida98}. Thus, \(S\) is also called the ``center generating'' function.

For a pure state \(\ket{\Psi}\), the Wigner function given by Eq.~\ref{eq:contweylfunction} simplifies to:
\begin{eqnarray}
  && {\Psi}_x(\bs x_p, \bs x_q) = \left(2 \pi \hbar\right)^{-n}\\
  && \int^\infty_{-\infty} \mbox{d} \bs \xi_q \, \Psi \left( \bs x_q + \frac{\bs \xi_q }{2} \right) {\Psi^*} \left( \bs x_q - \frac{\bs \xi_q}{2} \right) e^{-\frac{i}{\hbar} \bs \xi_q \cdot \bs x_p}. \nonumber
\end{eqnarray}

The center representation for quantum states immediately yields the well-known Wigner function for continuous systems. The chord representation is the symplectic Fourier transform of the Wigner function.  The center and chord representations are dual to each other, and are the Wigner and Wigner characteristic functions res[ectively. Identifying the Wigner functions with the center representation, and the center representation as dual to the chord representation motivates the development of both center (Wigner) and chord representations for discrete systems in Section~\ref{sec:discrete}.

\section{Path Integral Propagation in continuous Systems}
\label{sec:contpathintegral}

Propagation from one quantum state to another can be expressed in terms of the path integral formalism of the quantum propagator. For one degree of freedom, with an initial position \(q\) and final position \(q'\), evolving under the Hamiltonian \(H\) for time \(t\), the propagator is
\begin{equation}
  \expval{q}{e^{-i H t/\hbar}}{q'} = \int \mathcal{D}[q_t] \, \exp \left(\frac{i}{\hbar} G[q_t]\right)
  \label{eq:feynmanprop}
\end{equation}
where \(G[q_t]\) is the action of the trajectory \(q_t\), which starts at \(q\) and ends at \(q'\) a time \(t\) later~\cite{Feynman12,Schulman12}.

Eq.~\ref{eq:feynmanprop} can be reexpressed as a variational expansion around the set of classical trajectories (a set of measure zero) that start at  \(q\) and end at \(q'\) a time \(t\) later. This is an expansion in powers of \(\hbar\):
\begin{eqnarray}
  \label{eq:semiclassprop}
  && \expval{q'}{e^{-\frac{i}{\hbar}Ht}}{q} =\\
  && \sum_j^\text{cl. paths} \int \mathcal{D}[q_{tj}] \, e^{\frac{i}{\hbar} \left( G[q_{tj}] + \delta G[q_{tj}] + \frac{1}{2} \delta^2 G[q_{tj}] + \ldots \right) }, \nonumber
\end{eqnarray}
where \(\delta G[q_{tj}]\) denotes a functional variation of the paths \(q_{tj}\) and for classical paths \(\delta G[q_{tj}] = 0\)~(For further details we refer the reader to Section 10.3 of~\cite{Tannor07}).

Terminating Eq.~\ref{eq:semiclassprop} to first order in \(\hbar\) produces the position state representation of the van Vleck-Morette-Gutzwiller propagator~\cite{Van28,Morette51,Gutzwiller67}:
\begin{eqnarray}
  \label{eq:vVMG}
  && \expval{q'}{e^{-\frac{i}{\hbar}Ht}}{q} = \\
  && \sum_j\left( \frac{- \frac{\partial^2 G_{jt}(q,q')}{\partial q \partial q'}}{2 \pi i \hbar} \right)^{1/2} e^{i \frac{G_{jt}(q,q')}{\hbar}} + \mathcal{O}(\hbar^2).\nonumber
\end{eqnarray}
where the sum is over all classical paths that satisfy the boundary conditions.

In the center representation, for \(n\) degrees of freedom, the semiclassical propagator \(U_t(\bs x_p, \bs x_q)\) becomes~\cite{Almeida98}:
\begin{eqnarray}
  &&U_t(\bs x_p, \bs x_q) = \\
  && \sum_j \left\{ \det \left[ 1 + \frac{1}{2} \bs{\mathcal J} \frac{\partial^2 S_{tj}}{\partial \bs x^2} \right] \right\}^{\frac{1}{2}} e^{\frac{i}{\hbar} S_{tj}(\bs x_p, \bs x_q)} + \mathcal{O}(\hbar^2),\nonumber
  \label{eq:contcenterrepvVMG}
\end{eqnarray}
where \(S_{tj}(\bs x_p, \bs x_q)\) is the center generating function (or action) for the center \(\bs x = (\bs x_p, \bs x_q)\equiv\frac{1}{2}\left[(\bs p, \bs q)+(\bs p', \bs q')\right]\).

In general this is an underdetermined system of equations and there are an infinite number of classical trajectories that satisfy these conditions. The accuracy of adding them up as part of this semiclassical approximation is determined by how separated these trajectories are with respect to \(\hbar\)---the saddle-point condition for convergence of the method of steepest descents. However, some Hamiltonians exhibit a single saddle point contribution and are thus exact at order \(\hbar^1\)~\footnote{See Ehrenfest's Theorem, e.g. in~\cite{gottfried13}}:

\begin{lemma}There is only one classical trajectory \((\bs p, \bs q)\underset{t}{\rightarrow} (\bs p', \bs q')\) that satisfies the boundary conditions \((\bs x_p, \bs x_q) = \frac{1}{2}\left[(\bs p, \bs q)+(\bs p', \bs q')\right]\) and \(t\) under Hamiltonians that are harmonic in \(\bs p\) and \(\bs q\).\end{lemma}
\begin{proof}
  For a quadratic Hamiltonian, the diagonalized equations of motion for \(n\)-dimensional \((\bs p', \bs q')\) are of the form:
  \begin{eqnarray}
    p'_i &=& \alpha(t)_i p_i + \beta(t)_i q_i + \gamma(t)_i,\\
    q'_i &=& \delta(t)_i p_i + \epsilon(t)_i q_i + \eta(t)_i,
  \end{eqnarray}
  for \(i \in \{1, 2, \ldots, n\}\).
  Since \(t\) is known and \((\bs p, \bs q)\) can be written in terms of \((\bs p', \bs q')\) by using \((\bs x_p, \bs x_q)\), this brings the total number of linear equations to \(2n\) with \(2n\) unknowns and so there exists one unique solution.\qed
\end{proof}

Since the equations of motion for a harmonic Hamiltonian are linear, we can write their solutions as:
\begin{equation}
  \left( \begin{array}{c}\bs p'\\ \bs q'\end{array} \right) = \bs{\mathcal M}_t \left[ \left( \begin{array}{c}\bs p\\ \bs q \end{array}\right) + \frac{1}{2} \bs \alpha_t \right] + \frac{1}{2} \bs \alpha_t,
  \label{eq:quadmap}
\end{equation}
where \(\mathcal {\bs \alpha}_t\) is an \(n\)-vector and \(\bs{\mathcal M}_t\) is an \(n\times n\) symplectic matrix, both with entries in \(\mathbb{R}\).
In this case, the center generating function \(S_t(\bs x_p, \bs x_q)\) is also quadratic, in particular
\begin{equation}
  \label{eq:quadcentgenfunction}
  S_t(\bs x_p, \bs x_q) =  \bs \alpha_t^T \bs {\mathcal J} \left(\begin{array}{c}\bs x_p\\ \bs x_q\end{array}\right) + (\bs x_p, \bs x_q) \bs{\mathcal B}_t \left(\begin{array}{c}\bs x_p\\ \bs x_q\end{array}\right),
\end{equation}
where \(\bs{\mathcal B}_t\) is a real symmetric \(n\times n\) matrix that is related to \(\bs{\mathcal M}_t\) by the Cayley parameterization of \(\bs{\mathcal M}_t\)~\cite{Golub12}:
\begin{equation}
  \label{eq:cayleyparam}
  \bs {\mathcal J} \bs{\mathcal B}_t = \left( 1 + \bs {\mathcal M}_t \right)^{-1} \left( 1 - \bs {\mathcal M}_t \right) = \left( 1 - \bs {\mathcal M}_t \right)  \left( 1 + \bs {\mathcal M}_t \right)^{-1}.
\end{equation}

Since one classical trajectory contribution is sufficient in this case, if the overall phase of the propagated state is not important, then the expansion w.r.t.~\(\hbar\) in Eq.~\ref{eq:semiclassprop} can be truncated at order \(\hbar^0\). Dropping terms that are higher order than \(\hbar^0\) and ignoring phase is equivalent to propagating the classical density \(\rho(\bs x)\) corresponding to the \((\bs p, \bs q) \)-manifold, under the harmonic Hamiltonian and determining its overlap with the \((\bs p', \bs q')\)-manifold after time \(t\). Such a treatment under a harmonic Hamiltonian results in just the absolute value of the prefactor of Eq.~\ref{eq:contcenterrepvVMG}: \(\left| \det \left[ 1 + \frac{1}{2} \bs{\mathcal J} \frac{\partial^2 S_{tj}}{\partial \bs x^2} \right] \right|^{\frac{1}{2}}\). Indeed, this was van Vleck's discovery before quantum mechanics was formalized~\cite{Van28}. The relative phases of different classical contributions are no longer a concern and the higher order terms only weigh such contributions appropriately.

Here we are interested in propagating between Gaussian states in the center representation. In continuous systems, a Gaussian state in \(n\) dimensions can be defined as:
\begin{equation}
\label{eq:Psi}
  \Psi_\beta(\bs q) = \left[\pi^{-n} \det\left( \Re {\bs \Sigma}_\beta \right) \right]^{\frac{1}{4}} \exp \left( \varphi \right),
\end{equation}
where
\begin{equation}
  \varphi = \frac{i}{\hbar} \bs p_\beta \cdot \left( \bs q - {\bs q}_\beta \right) - \frac{1}{2} \left( \bs q - {\bs q}_\beta \right)^T {\bs \Sigma}_\beta \left( \bs q - \bs q_\beta \right).
\end{equation}
\(\bs q_\beta \in \mathbb{R}^n\) is the central position, \(\bs p_\beta \in \mathbb{R}^n\) is the central momentum, and \(\bs \Sigma_\beta\) is a symmetric \(n\times n\) matrix where \(\Re \bs \Sigma_\beta\) is proportional to the spread of the Gaussian and \(\Im \bs \Sigma_\beta\) captures \(p\)-\(q\) correlation.

This state describes momentum states (\(\delta(\bs p - \bs p_\beta)\) in momentum representation) when \(\bs \Sigma_\beta \rightarrow 0\) and position states \(\delta(\bs q- \bs q_\beta)\) when \(\bs \Sigma_\beta \rightarrow \infty\). Rotations between these two cases corresponds to \(\Re \bs \Sigma_\beta = 0\) and \(\Im \bs \Sigma_\beta \ne 0\).

Gaussians remain Gaussians under evolution by a harmonic Hamiltonian, even if it is time-dependent. This can be shown by simply making the ansatz that the state remains a Gaussian and then solving for its time-dependent \(\bs \Sigma_\beta\), \(\bs p_\beta\), \(\bs q_\beta\) and phase from the time-dependent Schr\"odinger equation~\cite{Tannor07}, or just by applying the analytically known Feynman path integral, which is equivalent to the van Vleck path integral, to a Gaussian~\cite{Heller91}.

Moreover, with the propagator in the center representation known to only have have one saddle-point contribution for a harmonic Hamiltonian, it is fairly straight-forward to show that this is also true for its coherent state representation (that is, taking a Gaussian to another Gaussian). Applying the propagator to an initial and final Gaussian in the center representation
\begin{eqnarray}
  && \left[ \ketbra{\Psi_\beta}{\Psi_\beta} U_t \ketbra{\Psi_\alpha}{\Psi_\alpha} \right]_x (\bs x) = \\
  && \left( \pi \hbar \right)^{-3n} \int^\infty_{-\infty} \mbox d \bs x_1 \int^\infty_{-\infty} \mbox d \bs x_2 \, U_t(\bs x_1 + \bs x_2- \bs x) \nonumber\\
  && \qquad \qquad \qquad \times {\Psi_\beta}_x(\bs x_2) {\Psi_\alpha}_x (\bs x_1) \nonumber\\
  && \qquad \qquad \qquad \times e^{2\frac{i}{\hbar}(\bs x_1^T \bs{\mathcal J} {\bs x_2} + \bs x_2^T \bs{\mathcal J} {\bs x} + \bs x^T \bs{\mathcal J} {\bs x_1})}, \nonumber
\end{eqnarray}
we see that since \(U_t\) is a Gaussian from Eq.~\ref{eq:contcenterrepvVMG} (\(S_{tj}\) is quadratic for harmonic Hamiltonians) and since the Wigner representations of the Gaussians, \({\Psi_\alpha}_x\) and \({\Psi_\beta}_x\), are also known to be Gaussians, the full integral in the above equation is a Gaussian integral and thus evaluates to produce a Gaussian with a prefactor. This is equivalent to evaluating the integral by the method of steepest descents which finds the saddle points to be the points that satisfy \(\partder{\phi}{\bs x} = 0\) where \(\phi\) is the phase of the integrand's argument. Since this argument is quadratic, its first derivative is linear and so again there is only one unique saddle point.

Indeed, such an evaluation produces the coherent state representation of the vVMG propagator~\cite{Tomsovic15}. Just as we found with the center representation, the absolute value of its prefactor corresponds to the order \(\hbar^0\) term.

As a consequence of this single contribution at order \(\hbar^0\), it follows that the Wigner function of a state, \(\Psi_x(\bs x)\), evolves under the operator \(\hat V\) with an underlying harmonic Hamiltonian by \(\Psi_x(\bs{\mathcal M}_{\hat V}(\bs x + \bs \alpha_{\hat V}/2) + \bs \alpha_{\hat V}/2)\), where \(\bs{\mathcal M}_{\hat V}\) is the symplectic matric and \(\bs \alpha_{\hat V}\) is the translation vector associated with \(\hat V\)'s action~\cite{Rivas99}. 

Before proceeding to the discrete case, we note that the center representation that we have defined allows for a particularly simple way to express how far the path integral treatment must be expanded in \(\hbar\) in order to describe \emph{any} unitary propagation (not necessarily harmonic) in continuous quantum mechanics.

Reflections and translations can also be described by truncating Eq.~\ref{eq:contcenterrepvVMG} at order \(\hbar^1\) (or \(\hbar^0\) if overall phase isn't important) since they correspond to evolution under a harmonic Hamiltonian. In particular, translations are displacements along a chord \(\bs \xi\) and so have Hamiltonians \(H \propto \bs \xi_q \cdot \bs p - \bs \xi_p \cdot \bs q\). Reflections are symplectic rotations around a center \(\bs x\) and so have Hamiltonians \(H \propto \frac{\pi}{4}\left[ (\bs p- \bs x_p)^2 + (\bs q- \bs x_q)^2 \right]\).

  From Eq.~\ref{eq:contsupofreflections} we see that any operator can be expressed as an infinite Riemann sum of reflections. Therefore, since reflections are fully described by a truncation at order \(\hbar^1\), it follows that an infinite Riemann sum of path integral solutions truncated at order \(\hbar^1\) can describe any unitary evolution. The same statement can be made by considering the chord representation in terms of translations.

  Hence, quantum propagation in continuous systems can be fully treated by an infinite sum of contributions from a path integral approach truncated at order \(\hbar^1\).

  As an aside, in general this infinite sum isn't convergent and so it is often more useful to consider reformulations that involve a sum with a finite number of contributions. One way to do this is to apply the method of steepest descents \emph{directly} on the operator of interest and use the area between saddle points as the metric to determine the order of \(\hbar\) necessary, instead of dealing with an infinity of reflections (or translations). This results in the semiclassical propagator already presented, but associated with the full Hamiltonian instead of a sum of reflection Hamiltonians.

  In summary, we have explained why propagation between Gaussian states under Hamiltonians that are harmonic is simulable classically (i.e. up to order \(\hbar^0\)) in continuous systems. We will see that the same situation holds in discrete systems for stabilizer states, with the additional restriction that the propagation takes the phase space points, which are now discrete, to themselves.

\section{Discrete Center-Chord Representation}
\label{sec:discrete}

We now proceed to the discrete case and introduce the center-chord formalism for these systems. It will be useful for us to define a pair of conjugate degrees of freedom \(p\) and \(q\) for discrete systems. Unfortunately, this isn't as straight-forward as in the continuous case, since the usual canonical commutation relations cannot hold in a finite-dimensional Hilbert space where the operators are bounded (since \(\Tr [\hat p, \hat q] = 0\)).

We begin in one degree of freedom. We label the computational basis for our system by \(n \in {0, 1, \ldots, d-1}\), for \(d\) odd and we assume that \(d\) is odd for the rest of this paper. We identify the discrete position basis with the computational basis and define the ``boost'' operator as diagonal in this basis:
\begin{equation}
  \hat {Z}^{\delta p} \ket{n} \equiv \omega^{n \delta p} \ket{n},
\end{equation}
where \(\omega\) will be defined below.

We define the normalized discrete Fourier transform operator to be equivalent to the Hadamard gate:
\begin{equation}
  \hat {F} = \frac{1}{\sqrt{d}}\sum_{m,n \in \mathbb{Z}/d\mathbb{Z}} \omega^{m n} \ket{m}\bra{n}.
  \label{eq:hadamard}
\end{equation}
This allows us to define the Fourier transform of \(\hat {Z}\):
\begin{equation}
   \hat {X} \equiv \hat {F}^\dagger \hat {Z} \hat {F}
\end{equation}
Again, as before, we call \(\hat {X}\) the ``shift'' operator since
\begin{equation}
  \hat {X}^{\delta q} \ket{n} \equiv \ket{n\oplus\delta q},
\end{equation}
where \(\oplus\) denotes mod-\(d\) integer addition.
It follows that the Weyl relation holds again:
\begin{equation}
  \hat {Z} \hat {X} = \omega \hat {X} \hat {Z}.
\end{equation}

The group generated by \(\hat {Z}\) and \(\hat {X}\) has a \(d\)-dimensional irreducible representation only if \(\omega^d=1\) for odd \(d\). Equivalently, there are only reflections relating any two phase space points on the Weyl phase space ``grid'' if \(d\) is odd~\cite{Rivas00}. We take \(\omega \equiv \omega(d) = e^{2 \pi i/d}\)~\cite{Sun92}. This was introduced by Weyl~\cite{Weyl32}. 

Note that this means that \(\hbar = \frac{d}{2 \pi}\) or \(h = d\). This means that the classical regime is most closely reached when the dimensionality of the system is reduced (\(d \rightarrow 0\)) and thus the most ``classical'' system we can consider here is a qutrit (since we keep \(d\) odd and greater than one). This is the opposite limit considered by many other approaches where the classical regime is reached when \(d \rightarrow \infty\).

One way of interpreting the classical limit in this paper is by considering \(h\) to be equal to the inverse of the density of states in phase space (i.e. in a Wigner unit cell). As \(\hbar \rightarrow 0\) phase space area decreases as \(\hbar^2\) but the number of states only decreases as \(\hbar\) leading to an overall density increase of \(\hbar\). This agrees with the notion that states should become point particles of fixed mass in the classical limit.

By analogy with continuous finite translation operators, we reexpress the shift \(\hat {X}\) and boost \(\hat {Z}\) operators in terms of conjugate \(\hat { p}\) and \(\hat { q}\) operators:\\
\begin{equation}   \hat {Z} \ket{n} = e^{\frac{2 \pi i }{d}\hat q} \ket{n} = e^{\frac{2 \pi i}{d} n} \ket{n}, \end{equation}
and
\begin{equation}   \hat {X} \ket{n} = e^{-\frac{2 \pi i }{d} \hat p} \ket {n} = \ket{n\oplus1}. \end{equation}
Hence, in the diagonal ``position'' representation for \(\hat Z\):
\begin{equation}
  \hat {Z} = \left( \begin{array}{cccccc} 1 & 0 & 0 & \cdots & 0\\
  0 & e^{\frac{2 \pi i}{d}} & 0 & \cdots & 0\\
  0 & 0 & e^{\frac{4 \pi i}{d}} & \cdots & 0\\
  \vdots & \vdots & \vdots & \ddots & \vdots\\
  0 & 0 & 0 & 0 & e^{\frac{2 (d-1) \pi i}{d}} \end{array} \right),
\end{equation}
and
\begin{equation}
  \hat {X} = \left( \begin{array}{cccccc} 0 & 0 & \cdots & 0 & 1\\
  1 & 0 & \cdots & 0 & 0\\
  0 & 1 & 0 & \cdots & 0\\
  \vdots & \ddots & \ddots & \ddots & \vdots\\
  0 & \cdots & 0 & 1 & 0 \end{array} \right).
\end{equation}

Thus,
\begin{equation}
  \hat {q} = \frac{d}{2 \pi i} \log \hat {Z} = \sum_{n \in \mathbb{Z}/d\mathbb{Z}} n \ket{n} \bra{n},
\end{equation}
and
\begin{equation}
  \hat { p} = \hat {F}^\dagger \hat { q} \hat {F}.
\end{equation}

Therefore, we can interpret the operators \(\hat { p}\) and \(\hat { q}\) as a conjugate pair similar to conjugate momenta and position in the continuous case. However, they differ from the latter in that they only obey the weaker \emph{group} commutation relation
\begin{equation}
  e^{i j \hat { q}/\hbar}e^{i k \hat { p}/\hbar} e^{-i j \hat { q}/\hbar}e^{-i k \hat { p}/\hbar} = e^{-i jk /\hbar} \hat {\mathbb I}.
  \label{eq:groupcommrel}
\end{equation}
This corresponds to the usual canonical commutation relation for \(p\) and \(q\)'s algebra at the origin of the Lie group (\(j = k = 0\)); expanding both sides of Eq.~\ref{eq:groupcommrel} to first order in \(p\) and \(q\) yields the usual canonical relation.

We proceed to introduce the Weyl representation of operators and states in discrete Hilbert spaces with odd dimension \(d\) and \(n\) degrees of freedom~\cite{Wootters87,Wootters03,Gibbons04}. The generalized phase space translation operator (the Weyl operator) is defined as a product of the shift and boost with a phase appropriate to the \(d\)-dimensional space:
\begin{equation}
\hat {T}(\bs \xi_p, \bs \xi_q) = e^{-i \frac{\pi}{d} \bs \xi_p \cdot \bs \xi_q} \hat {Z}^{ \bs \xi_p} \hat {X}^{ \bs \xi_q},
\end{equation}
where \(\bs \xi \equiv (\bs \xi_p, \bs \xi_q) \in (\mathbb{Z} / d \mathbb{Z})^{2n}\) and form a discrete ``web'' or ``grid'' of chords. They are a discrete subset of the continous chords we considered in the infinite-dimensional context in Section~\ref{sec:contcenterchord} and their finite number is an important consequence of the discretization of the continuous Weyl formalism.

Again, an operator \(\hat {A}\) can be expressed as a linear combination of translations:
\begin{equation}
  \hat {A} = d^{-n} \sum_{\substack{\bs \xi_p, \bs \xi_q \in \\ (\mathbb{Z} / d \mathbb{Z})^{ n}}} A_\xi(\bs \xi_p, \bs \xi_q) \hat {T}(\bs \xi_p, \bs \xi_q),
\end{equation}
where the weights are the chord representation of the function \(\hat {A}\):
\begin{equation}
  {A}_\xi(\bs \xi_p, \bs \xi_q) = d^{-n} \Tr \left( \hat {T}(\bs \xi_p, \bs \xi_q)^\dagger \hat {A} \right).
\end{equation}
When applied to a state \(\hat \rho\), this is also called the ``characteristic function'' of \(\hat \rho\)~\cite{Ferrie11}.

As before, the center representation, based on reflections instead of translations, requires an appropriately defined reflection operator. We can define the discrete reflection operator \(\hat {R}\) as the symplectic Fourier transform of the discrete translation operator we just introduced:
\begin{equation}
\hat {R}(\bs x_p, \bs x_q) = d^{-n} \sum_{\substack{\bs \xi_p, \bs \xi_q \in \\ (\mathbb{Z} / d \mathbb{Z})^{ n}}} e^{\frac{2 \pi i}{d} (\bs \xi_p, \bs \xi_q) \bs{\mathcal J} (\bs x_p, \bs x_q)^T} \hat {T}(\bs \xi_p, \bs \xi_q).
\end{equation}

With this in hand, we can now express a finite-dimensional operator \(\hat {A}\) as a superposition of reflections:
\begin{equation}
  \hat {A} = d^{-n} \sum_{\substack{\bs x_p, \bs x_q \in \\ (\mathbb{Z} / d \mathbb{Z})^{ n}}} {A}_x(\bs x_p, \bs x_q) \hat {R}( \bs x_p, \bs x_q ),
  \label{eq:discretesupofreflections}
\end{equation}
where
\begin{equation}
  {A}_x(\bs x_p, \bs x_q) = d^{-n} \Tr \left( \hat {R}(\bs x_p, \bs x_q)^\dagger \hat {A} \right).
  \label{eq:weylfunction}
\end{equation}
\(\bs x \equiv (\bs x_p, \bs x_q) \in (\mathbb{Z} / d \mathbb{Z})^{2n}\) are centers or Weyl phase space points and, like their \((\bs \xi_p, \bs \xi_q)\) brethren, form a discrete subgrid of the continuous Weyl phase space points considered in Section~\ref{sec:contcenterchord}.

Again, the center representation is of particular interest to us because for unitary gates \(\hat A\) we can rewrite the components \({A}_x(\bs x_p, \bs x_q)\) as:
\begin{equation}
{A}_x(\bs x_p, \bs x_q) = \exp \left[\frac{i}{\hbar} S(\bs x_p, \bs x_q)\right]
\end{equation}
where \(S(\bs x_p, \bs x_q)\) is the argument to the exponential, and is equivalent to the action of the operator in center representation (the center generating function).

Aside from Eq.~\ref{eq:weylfunction}, the center representation of a state \(\hat \rho\) can also be directly defined as the symplectic Fourier transform of its chord representation, \(\rho_\xi\)~\cite{Gross06}:
\begin{equation}
  \rho_x(\bs x_p, \bs x_q) = d^{-n} \sum_{\substack{\bs \xi_p, \bs \xi_q \in \\ (\mathbb{Z} / d \mathbb{Z})^{ n}}} e^{\frac{2 \pi i}{d} (\bs \xi_p, \bs \xi_q) \bs{\mathcal J} (\bs x_p, \bs x_q)^T} \rho_\xi(\bs \xi_p,\bs \xi_q).
  \label{eq:weylfunction2}
\end{equation}

We note again that for a pure state \(\ket{\Psi}\), the Wigner function from Eqs.~\ref{eq:weylfunction} and~\ref{eq:weylfunction2} simplifies to:
\begin{widetext}
\begin{eqnarray}
  \label{eq:weylfunctionpurestatediscrete}
  {\Psi}_x(\bs x_p, \bs x_q) &=& d^{-n} \sum_{\substack{\bs \xi_q \in \\(\mathbb{Z} / d \mathbb{Z})^{ n}}} e^{-\frac{2 \pi i}{d} \bs \xi_q \cdot \bs x_p} \Psi \left( \bs x_q + \frac{(d+1) \bs \xi_q }{2} \right) {\Psi^*} \left( \bs x_q - \frac{(d+1) \bs \xi_q}{2} \right).
\end{eqnarray}
\end{widetext}

It may be clear from this short presentation that the chord and center representations are dual to each other. A thorough review of this subject can be found in~\cite{Almeida98}.

\section{Path Integral Propagation in Discrete Systems}
\label{sec:discretesemi}

Rivas and Almeida~\cite{Rivas99} found that the continuous infinite-dimensional vVMG propagator can be extended to finite Hilbert space by simply projecting it onto its finite phase space tori. This produces:
  \begin{eqnarray}
    \label{eq:discretecenterrepvVMG}
    && {U}_t(\bs x) = \\
    && \left< \sum_j \left\{ \det \left[ 1 + \frac{1}{2} \bs{\mathcal J} \frac{\partial^2 S_{tj}}{\partial \bs{x}^2} \right] \right\}^{\frac{1}{2}} e^{\frac{i}{\hbar} S_{tj}(\bs x)} e^{i \theta_k} \right>_k + \mathcal{O}(\hbar^2),\nonumber
  \end{eqnarray}
where an additional average must be taken over the \(k\) center points that are equivalent because of the periodic boundary conditions. Maintaining periodicity requires that they accrue a phase \(\theta_k\)~\footnote{See Eq.~\(5.18\) in~\cite{Rivas99} for a definition of the phase.}. The derivative \(\frac{\partial^2 S_{tj}}{\partial \bs {x}^2}\) is performed over the continuous function \(S_{tj}\) defined after Eq.~\ref{eq:contcenterrepvVMG}, but only evaluated at the discrete Weyl phase space points \(\bs x \equiv (\bs x_p, \bs x_q)\in \left(\mathbb{Z}/d\mathbb{Z}\right)^{2n}\). The prefactor can also be reexpressed:
\begin{equation}
  \left\{ \det \left[ 1 + \frac{1}{2} \bs{\mathcal J} \frac{\partial^2 S_{tj}}{\partial \bs{x}^2} \right] \right\}^{\frac{1}{2}} = \left\{ 2^d \det \left[ 1 + \bs{\mathcal M_{tj}} \right] \right\}^{-\frac{1}{2}},
\end{equation}
which is perhaps more pleasing in the discrete case as it does not involve a continuous derivative.

As in the continuous case, for a harmonic Hamiltonian \(H(\bs p, \bs q)\), the center generating function \(S(\bs x_p, \bs x_q)\) is equal to \(\bs{\alpha}^T \bs{\mathcal J} \bs x + \bs{x}^T \bs{\mathcal B} \bs x\) where
Eq.~\ref{eq:cayleyparam} and Eq.~\ref{eq:quadmap} hold. Moreover, if the Hamiltonian takes Weyl phase space points to themselves, then by the same equations it follows that \(\bs{\mathcal M}\) and \(\bs \alpha\) must have integer entries.

This implies that for \(\bs m, \bs n \in \mathbb{Z}^{ n}\),
\begin{eqnarray}
  && \bs{\mathcal M} \left(\begin{array}{c}\bs p+ \bs m d  + \bs \alpha_p/2\\\bs q + \bs n d + \bs \alpha_q/2\end{array}\right) + \left( \begin{array}{c} \bs \alpha_p/2\\ \bs \alpha_q/2 \end{array} \right) \nonumber\\
  \label{eq:harmonicM}
  &=& \bs{\mathcal M} \left(\begin{array}{c}\bs p + \bs \alpha_p/2 \\ \bs q + \bs \alpha_q/2 \end{array}\right) + \left( \begin{array}{c} \bs \alpha_p/2\\ \bs \alpha_q/2 \end{array} \right) + d \bs{\mathcal M}\left(\begin{array}{c}\bs m\\ \bs n\end{array}\right) \\
  &=& \left(\begin{array}{c}\bs p'\\ \bs q'\end{array}\right) \mod \bs d. \nonumber
\end{eqnarray}
Therefore, phase space points \((\bs p, \bs q)\) that lie on Weyl phase space points go to the equivalent Weyl phase space points \((\bs p', \bs q')\).

Moreover, again if \(\bs m, \bs n \in \mathbb{Z}^n\),
\begin{eqnarray}
  && S(\bs x_p + \bs m d, \bs x_q + \bs n d) \nonumber\\
  &=& \left( \begin{array}{c}\bs x_p+ \bs m d\\ \bs x_q+ \bs n d\end{array} \right)^T \bs A \left( \begin{array}{c}\bs x_p+ \bs m d\\ \bs x_q+ \bs n d\end{array} \right) \nonumber\\
      && + \bs b \cdot \left( \begin{array}{c}\bs x_p+ \bs m d\\ \bs x_q+ \bs n d\end{array} \right)\\
  &=& \left( \begin{array}{c}\bs x_p\\ \bs x_q\end{array} \right)^T \bs A \left( \begin{array}{c}\bs x_p\\ \bs x_q\end{array} \right) + \bs b \cdot \left( \begin{array}{c}\bs x_p\\ \bs x_q\end{array} \right) \nonumber\\
              && + d \left[ 2 \left( \begin{array}{c}\bs x_p\\ \bs x_q\end{array} \right)^T \bs A \left( \begin{array}{c}\bs m\\ \bs n\end{array} \right) \right.\nonumber\\
                    && \qquad \left.+ d \left( \begin{array}{c}\bs m\\ \bs n\end{array} \right)^T \bs A \left( \begin{array}{c}\bs m\\ \bs n\end{array} \right) + \bs b \cdot \left( \begin{array}{c}\bs m\\ \bs n\end{array} \right)\right] \nonumber\\
  \label{eq:harmonicS}
  &=& S(\bs x_p, \bs x_q) \mod d, \nonumber
\end{eqnarray}
for some symmetric \(\bs A \in \mathbb{Z}^{n\times n}\) and \(\bs b \in \mathbb{Z}^{n}\). Therefore, these equivalent trajectories also have equivalent actions (since the action is multiplied by \(\frac{2 \pi i}{d}\) and exponentiated).

Hence, there is only one term to the sum in Eq.~\ref{eq:discretecenterrepvVMG}. Moreover,~\cite{Rivas00} showed that the sum over the phases \(\theta_k\) produces only a global phase that can be factored out. Therefore, if we can neglect the overall phase,
\begin{equation}
  {U}_t(\bs x) = \left| 2^d \det \left[ 1 + \bs{\mathcal M} \right] \right|^{\frac{1}{2}},
  \label{eq:discretesemiproplowestorder}
\end{equation}
where the classical trajectories whose centers are \((\bs x_p, \bs x_q)\) satisfy the periodic boundary conditions.

As in the continuous case, we point out that this means that translations and reflections are fully captured by a path integral treatment that is truncated at order \(\hbar^1\) (or order \(\hbar^0\) if their overall phase isn't important) because their Hamiltonians are harmonic, but in the discrete case there is an additional requirement that they are evaluated at chords/centers that take Weyl phase space points to themselves.

Just as in the continuous case, the single contribution at order \(\hbar^0\) implies that the propagator of the Wigner function of states, \(\ketbra{\Psi}{\Psi}_x(\bs x)\) under gates \(\hat V\) with underlying harmonic Hamiltonians is captured by \(\ketbra{\Psi}{\Psi}_x(\bs{\mathcal M}_{\hat V} (\bs x + \bs \alpha_{\hat V}/2) + \bs \alpha_{\hat V}/2)\) for \(\bs{\mathcal M}_{\hat V}\) and \(\bs \alpha_{\hat V}\) associated with \(\hat V\).

\section{Stabilizer Group}
\label{sec:stabilizergroup}

Here we will show that the Hamiltonians corresponding to Clifford gates are harmonic and take Weyl phase space points to themselves. Thus they can be captured by only the single contribution of Eq.~\ref{eq:discretesemiproplowestorder} at lowest order in \(\hbar\). This then implies that stabilizer states can also be propagated to each other by Clifford gates with only a single contribution to the sum in Eq.~\ref{eq:discretecenterrepvVMG}.

The Clifford gate set of interest can be defined by three generators: a single qudit Hadamard gate \(\hat{F}\) and phase shift gate \(\hat{P}\), as well as the two qudit controlled-not gate \(\hat{C}\). We examine each of these in turn.

\subsection{Hadamard Gate}
The Hadamard gate was defined in Eq.~\ref{eq:hadamard} and is a rotation by \(\frac{\pi}{2}\) in phase space counter-clockwise. Hence, for one qudit, it can be written as the map in Eq.~\ref{eq:harmonicM} where
\begin{equation}
  \label{eq:stabmathad}
\bs{\mathcal M}_{\hat {F}} = \left( \begin{array}{cc} 0 & 1\\ -1 & 0 \end{array} \right),
\end{equation}
and \(\bs \alpha_{\hat {F}} = (0,0)\). We have set \(t=1\) and drop it from the subscripts from now on.
Since \(\bs \alpha\) is vanishing and \(\bs{\mathcal M}\) has integer entries, this is a cat map and such maps have been shown to correspond to Hamiltonians~\cite{Keating91}
\begin{eqnarray}
  \label{eq:hadhamiltonian}
  &&H(p,q) =\\
  &&f(\Tr \bs{\mathcal M} ) \left[ \mathcal M_{12} p^2 - \mathcal M_{21} q^2 + \left(\mathcal M_{11} - \mathcal M_{22} \right) pq \right],\nonumber
\end{eqnarray}
where
\begin{equation}
  f(x) = \frac{\sinh^{-1}(\frac{1}{2}\sqrt{x^2-4})}{\sqrt{x^2-4}}.
\end{equation}
For the Hadamard \(\bs {\mathcal M}_{\hat {F}}\) this corresponds to \(H_{\hat {F}} = \frac{\pi}{4} (p^2 + q^2)\), a harmonic oscillator. The center generating function \(S(x_p, x_q)\) is thus \((x_p, x_q ) \bs{\mathcal B} (x_p, x_q)^T\) and solving Eq.~\ref{eq:cayleyparam} finds for the one-qudit Hadamard,
\begin{equation}
 \bs{\mathcal{B}}_{\hat {F}} = \left(\begin{array}{cc}1 & 0\\ 0 & 1 \end{array} \right).
\end{equation}
Thus, \(S_{\hat {F}}(x_p, x_q) = x_p^2 + x_q^2\). Indeed, applying Eq.~\ref{eq:weylfunction} to Eq.~\ref{eq:hadamard} reveals that the Hadamard's center function (up to a phase) is:
\begin{equation}
  {F}_x(x_p,x_q) = e^{\frac{2 \pi i}{d}(x_p^2+x_q^2)}.
  \label{eq:hadamardcenterfn}
\end{equation}

Eq.~\ref{eq:stabmathad} shows how to map Weyl phase space to Weyl phase space under the Hadamard transformation. Furthermore this map is pointwise, which implies the quadratic form of the center generating function obtained in Eq.~\ref{eq:hadamardcenterfn}.

\subsection{Phase Shift Gate}
The phase shift gate can be generalized to odd \(d\)-dimensions~\cite{Gottesman99} by setting it to:
\begin{equation}
  \hat{P} = \sum_{j \in \mathbb{Z}/d \mathbb{Z}} \omega^{\frac{(j-1)j}{2}} \ketbra{j}{j}.
  \label{eq:phaseshift}
\end{equation}
Examining its effect on stabilizer states, it is clear that it is a \(q\)-shear in phase space from an origin displaced by \(\frac{d-1}{2}\equiv-\frac{1}{2}\) to the right. This can be expressed as the map in Eq.~\ref{eq:harmonicM} with
\begin{equation}
\bs{\mathcal M}_{\hat {P}} = \left( \begin{array}{cc} 1 & 1\\ 0 & 1 \end{array} \right),
\end{equation}
and \(\bs \alpha_{\hat {P}} = \left( -\frac{1}{2}, 0 \right)\).

This corresponds to
\begin{equation}
\bs{\mathcal B}_{\hat {P}} = \left( \begin{array}{cc} 0 & 0\\ 0 & \frac{1}{2} \end{array} \right).
\end{equation}

Solving Eq.~\ref{eq:quadcentgenfunction} with this \(\bs{\mathcal B}_{\hat {P}}\) and \(\bs{\alpha}_{\hat {P}}\) reveals that \(S_{\hat {P}}(x_p, x_q) = -\frac{1}{2} x_q + \frac{1}{2} x_q^2\). Again, this agrees with the argument of the center representation of the phase-shift gate obtained by applying Eq.~\ref{eq:weylfunction} to Eq.~\ref{eq:phaseshift}:
\begin{equation}
{P}_x(x_p,x_q) = e^{\frac{2 \pi i}{d} \frac{1}{2} (-x_q + x_q^2)}.
\end{equation}
Discretization the equations of motion for harmonic evolution for unit timesteps leads to:
\begin{equation}
  \left(\begin{array}{c}\bs p'\\ \bs q'\end{array}\right) = \left(\begin{array}{c}\bs p\\ \bs q\end{array}\right) + \bs{\mathcal J} \left(\partder{H}{\bs p}, \partder{H}{\bs q}\right)^T,
  \label{eq:harmonicevol}
\end{equation}
where the last derivative is on the continuous function \(H\), but only evaluated on the discrete Weyl phase space points. It follows that
\begin{equation}
  \label{eq:phaseshifthamiltonian}
  H_{\hat {P}} = -\frac{d+1}{2} q^2 + \frac{d+1}{2} q.
\end{equation}

We have obtained the Hamiltonian for the phase-shift gate by a different procedure than that used for the Hadamard where we appealed to the result given in Eq.~\ref{eq:hadhamiltonian} for quantum cat maps. However, as the phase-shift gate is a quantum cat map as well, we could have obtained Eq.~\ref{eq:phaseshifthamiltonian} in this manner. Similarly, the approach we used to find the phase-shift Hamiltonian by discretizing time in Eq.~\ref{eq:harmonicevol} would work for the Hadamard gate but it is a bit more involved since the latter contains both \(p\)- and \(q\)-evolution. Nevertheless, this produces Eq.~\ref{eq:hadhamiltonian} as well. We presented both techniques for illustrative purposes.

\subsection{Controlled-Not Gate}

Lastly, the controlled-not gate can be generalized to \(d\)-dimensions~\cite{Gottesman99} by
\begin{equation}
  \hat{C} = \sum_{j,k \in \mathbb{Z}/d \mathbb{Z}} \ketbra{j,k \oplus j}{j,k}.
\end{equation}
It is clear that this translates the \(q\)-state of the second qudit by the \(q\)-state of the first qudit. As a result, as is evident by examining the gate's action on stabilizer states, the first qudit experiences an ``equal and opposite reaction'' force that kicks its momentum by the \(q\)-state of the second qudit. This is the phase space picture of the well-known fact that a CNOT examined in the \(\hat {X}\) basis has the control and target reversed with respect to the \(\hat Z\) basis. This can also seen by looking at its effect in the momentum (\(\hat {X}\)) basis:
\begin{equation}
  {\hat {F}}^\dagger \hat{C} \hat {F} = \sum_{j,k \in \mathbb{Z}/d \mathbb{Z}} \ketbra{j \ominus k,k}{j,k}.
\end{equation}
As a result, this gate is described by the map:
\begin{equation}
\bs{\mathcal M}_{\hat {C}} = \left( \begin{array}{ccccc} 1 & -1 & 0 & 0\\ 0 & 1 & 0 & 0\\ 0 & 0 & 1 & 0\\ 0 & 0 & 1 & 1\end{array}\right),
\end{equation}
and \(\bs \alpha_{\hat {C}} = (0,0,0,0)\).
This corresponds to
\begin{equation}
\bs{\mathcal B}_{\hat {C}} = \left( \begin{array}{ccccc} 0 & 0 & 0 & 0\\ 0 & 0 & -\frac{1}{2} & 0\\ 0 & -\frac{1}{2} & 0 & 0\\ 0 & 0 & 0 & 0\end{array}\right).
\end{equation}
Hence its center generating function \(S_{\hat {C}}(\bs x_p, \bs x_q) = -x_{p_2} x_{q_1}\). Again, this corresponds with the argument of the center representation of the controlled-not gate, which can be found to be:
\begin{equation}
{C}_x(\bs x_p, \bs x_q) = e^{-\frac{2 \pi i}{d} x_{q_1} x_{p_2}}.
\end{equation}
Therefore, this gate can be seen to be a bilinear \(p\)-\(q\) coupling between two qudits and corresponds to the Hamiltonian
\begin{equation}
  H_{\hat {C}} = p_1 q_2,
\end{equation}
as can be found from Eq.~\ref{eq:harmonicevol} again.

As a result, it is now clear that all the Clifford group gates have Hamiltonians that are harmonic and that take Weyl phase space points to themselves. Therefore, their propagation can be fully described by a truncation of the semiclassical propagator Eq.~\ref{eq:discretecenterrepvVMG} to order \(\hbar^0\) as in Eq.~\ref{eq:discretesemiproplowestorder} and they are manifestly classical in this sense.

To summarize the results of this section, using Eq.~\ref{eq:discretesupofreflections}, the Hadamard, phase shift, and CNOT gates can be written as:
\begin{equation}
  \hat{F} = d^{-2} \sum_{\substack{x_p,x_q, \\ \xi_p,\xi_q \in\\ \mathbb{Z}/d \mathbb{Z}}} e^{ -\frac{2 \pi i}{d} \left[ -(x_p^2 + x_q^2) - d \left(x_p \xi_p - x_q \xi_q \right) \right] } \hat {Z}^{\xi_p} \hat {X}^{\xi_q},
\end{equation}
\begin{equation}
  \hat{P} = d^{-2} \sum_{\substack{x_p,x_q, \\ \xi_p,\xi_q \in\\ \mathbb{Z}/d \mathbb{Z}}} e^{ -\frac{2 \pi i}{d} \left[ \frac{1}{2} (x_q - x_q^2) - d \left(x_p \xi_q - x_q \xi_p \right) \right] } \hat {Z}^{\xi_p} \hat {X}^{\xi_q},
\end{equation}
and
\begin{eqnarray}
  &&\hat{C} = \\
  && d^{-4} \sum_{\substack{\bs x_p, \bs x_q, \\ \bs \xi_p, \bs \xi_q \in \\ (\mathbb{Z}/d \mathbb{Z})^{ 2}}} e^{ -\frac{2 \pi i}{d} \left[ x_{q_1} x_{p_2} - d \left(\bs x_p \cdot \bs \xi_q - \bs x_q \cdot \bs \xi_p \right) \right] } \hat {Z}^{ \bs \xi_p} \hat {X}^{ \bs \xi_q}, \nonumber
\end{eqnarray}
(up to a phase). This form emphasizes their quadratic nature.

As for the continuous case, there exists a particularly simple way in discrete systems to see to what order in \(\hbar\) the path integral must be kept to handle unitary propagation beyond the Clifford group. We describe this in the next section.

We note that the center generating actions \(S(\bs x_p, \bs x_q)\) found here are related to the \(G(\bs q', \bs q)\) found by Penney \emph{et al}.~\cite{Penney16}, which are in terms of initial and final positions, by symmetrized Legendre transform~\cite{Almeida98}:
\begin{equation}
  \label{eq:actionfromsymmetrizedlegendretransform}
  G(\bs q', \bs q, t) = F\left(\frac{\bs q' + \bs q}{2}, \bs p(\bs q' - \bs q)\right),
\end{equation}
where the canonical generating function
\begin{equation}
  F\left(\frac{\bs q' + \bs q}{2}, \bs p\right) = S\left( \bs x_p = \bs p, \bs x_q = \frac{\bs q' + \bs q}{2} \right) + \bs p \cdot \left( \bs q' - \bs q \right),
\end{equation}
for \(\bs p (\bs q' - \bs q)\) given implicitly by \(\frac{\partial F}{\partial \bs p} = 0\).

Applying this to the actions we found reveals that
\begin{equation}
  G_{\hat {F}}(q', q, t) = q' q,
\end{equation}
\begin{equation}
  G_{\hat {P}}(q', q, t) = \frac{d+1}{2} (q^2 - q),
\end{equation}
and
\begin{equation}
  G_{\hat {C}}((q'_1, q'_2), (q_1, q_2), t) = 0,
\end{equation}
which is in agreement with~\cite{Penney16}.

\subsection{Classicality of Stabilizer States}

In this subsection we show that stabilizer states evolve to stabilizer states under Clifford gates, and that it is possible to describe this evolution classically. For odd \(d \ge 3\) the positivity of the Wigner representation implies that evolution of stabilizer states is non-contextual, and so here we are investigating in detail what this means in our semiclassical picture.

To begin, it is instructive to see the form stabilizer states take in the discrete position representation and in the center representation. Gross proved that~\cite{Gross06}:
\begin{theorem} \label{thm:stabstates}Let \(d\) be odd and \(\Psi \in L^2((\mathbb Z/d\mathbb Z)^n)\) be a state vector. The Wigner function of \(\Psi\) is non-negative if and only if \(\Psi\) is a stabilizer state. \end{theorem}
Gross also proved~\cite{Gross06} 
\begin{corollary}
  \label{cor:stabstates}
  Given that \(\Psi(\bs q) \ne 0\) \(\forall\, \bs q\), a vector \(\Psi\) is a stabilizer state if and only if it is of the form
\begin{equation}
  \Psi_{\theta_\beta,\eta_\beta}(\bs q) \propto \exp \left[\frac{2 \pi i}{d} \left( \bs q^T \bs \theta_\beta \bs q + \bs \eta_\beta \cdot \bs q \right) \right].
  \label{eq:stabstateposrep}
\end{equation}
where \(\bs \theta_\beta \in \left(\mathbb{Z}/d\mathbb{Z}\right)^{n\times n}\)  and \(\bs q, \bs \eta_\beta \in \left(\mathbb{Z}/d\mathbb{Z}\right)^n\).
\end{corollary}

Applying Eq.~\ref{eq:weylfunctionpurestatediscrete} to Eq.~\ref{eq:stabstateposrep}, the Wigner function of such maximally supported stabilizer states can be found to be:
\begin{eqnarray}
  && {\Psi_{\theta_\beta,\eta_\beta}}_x(\bs x_p, \bs x_q) \propto\\
  && d^{-n} \sum_{\substack{\bs \xi_q \in\\\left(\mathbb{Z}/d\mathbb{Z}\right)^{ n}}} \exp \left[\frac{2 \pi i}{d} \bs \xi_q \cdot \left( \bs \eta_\beta - \bs x_p + 2 \bs \theta_\beta \bs x_q \right)\right].\nonumber
\end{eqnarray}
Therefore, one finds that the Wigner function is the discrete Fourier sum equal to \(\delta_{\bs \eta_\beta - \bs x_p + 2\bs \theta_\beta \bs x_q}\). For \(\bs \theta_\beta = 0\) the state is a momentum state at \(\bs x_p\). Finite \(\bs \theta_\beta\) rotates that momentum state in phase space in ``steps'' such that it always lies along the discrete Weyl phase space points \((\bs x_p, \bs x_q) \in (\mathbb{Z}/d\mathbb{Z})^{2n}\).

This Gaussian expression only captures stabilizer states that are maximally supported in \(q\)-space. One may wonder what the stabilizer states that aren't maximally supported in \(q\)-space look like in Weyl phase space. Of course, it is possible that some may be maximally supported in \(p\)-space and so can be captured by the following corollary:
\begin{corollary}
  \label{cor:stabstates2}
  If \(\Psi(p) \ne 0\) for all \(p\)'s then there exists a \(\bs \theta_{\beta p} \in \left(\mathbb{Z}/d\mathbb{Z}\right)^{n\times n}\)  and an \(\bs \eta_{\beta p} \in \left(\mathbb{Z}/d\mathbb{Z}\right)^n\) such that
\begin{equation}
  \Psi_{\theta_{\beta p},\eta_{\beta p}}(\bs p) \propto \exp \left[\frac{2 \pi i}{d} \left( \bs p^T \bs \theta_{\beta p} \bs p + \bs \eta_{\beta p} \cdot \bs p \right) \right].
  \label{eq:stabstatemomrep}
\end{equation}
\end{corollary}
\begin{proof}
  This can be shown following the same methods employed by Gross~\cite{Gross06} but in the discrete \(p\)-basis.
\end{proof}

Unfortunately, it is easy to show that Corollary~\ref{cor:stabstates} and~\ref{cor:stabstates2} do not provide an expression for all stabilizer states (except for the odd prime \(d\) case, as we shall see shortly) as there exist stabilizer states for odd non-prime \(d\) that are not maximally supported in \(p\)- or \(q\)-space or any finite rotation between those two. To find an expression that encompasses all stabilizer states, we must turn to the Wigner function of stabilizer states.

An equivalent definition of stabilizer states on \(n\) qudits is given by states \(\hat {V} \underbrace{\ket{0} \otimes \cdots \otimes \ket{0}}_{n}\) where \(\hat {V}\) is a quantum circuit consisting of Clifford gates. We know that the Clifford circuits are generated by the \(\hat {P}\), \(\hat {F}\) and \(\hat {C}\) gates, and that the Wigner functions \(\Psi_x(\bs x)\) of stabilizer states propagate under \(\hat {V}\) as \(\Psi_x(\bs{\mathcal M}_{\hat {V}} (\bs x + \bs \alpha_{\hat {V}}/2) + \bs \alpha_{\hat {V}}/2)\), it follows that the Wigner function of stabilizer states is:
\begin{equation}
  \delta_{\bs \Phi_0 \cdot \bs{\mathcal M}_{\hat {V}} \cdot \bs x, \bs r_0},
  \label{eq:wignerfnofstabstateoddd}
\end{equation}
where \(\bs \Phi_0 = \left(\begin{array}{cc} 0 & 0\\ 0 & \mathbb{I}_n\end{array} \right)\) and \(\bs r_0 = (\bs 0, \bs 0)\). We have therefore proved the next theorem:
  \begin{theorem}
    \label{thm:wigfnofstabstates}
    The Wigner function \(\Psi_x(\bs x)\) of a stabilizer state for any odd \(d\) and \(n\) qudits is \(\delta_{\bs \Phi \cdot \bs x, \bs r}\) for \(2n \times 2n\) matrix \(\bs \Phi\) and \(2n\) vector \(\bs r\).
  \end{theorem}

As an aside, Theorem~\ref{thm:wigfnofstabstates} allows us to develop an all-encompassing Gaussian expression for stabilizer states for the restricted case that \(d\) is odd prime. In this case, the following Corollary shows that a ``mixed'' representation is always possible: where each degree of freedom is expressed in either the \(p\)- or \(q\)-basis:
\begin{corollary}
  \label{cor:stabstatemixedrep}
For odd prime \(d\), if \(\Psi\) is a stabilizer state for \(n\) qudits, then there always exists a mixed representation in position and momentum such that:
\begin{equation}
  \Psi_{\theta_{\beta \bs x},\eta_{\beta \bs x}}(\bs x) = \frac{1}{\sqrt{d}} \exp \left[\frac{2 \pi i}{d} \left( \bs x^T \bs \theta_{\beta \bs x} \bs x + \bs \eta_{\beta \bs x} \cdot \bs x \right) \right],
  \label{eq:stabstatemixedrep}
\end{equation}
where \(x_i\) can be either \(p_i\) or \(q_i\).
\end{corollary}
\begin{proof}
  We begin with a one-qudit case. We examine the equation specified by \(\bs \Phi \cdot \bs x = \bs r\):
  \begin{equation}
    \alpha q_1 + \beta p_1 = \gamma
  \end{equation}
  for \(\alpha\), \(\beta\), and \(\gamma \in \mathbb{Z}/d\mathbb{Z}\). If \(\alpha = 0\) then \(\Psi(q_1)\) is maximally supported and if \(\beta = 0\) then \(\Psi(p_1)\) is maximally supported since the equation specifies a line on \(\mathbb Z/d \mathbb Z\) in \(q_1\) and \(p_1\) respectively. If \(\alpha \ne 0\) and \(\beta \ne 0\) then the equation can be rewritten as
  \begin{equation}
    q_1 + (\beta/\alpha) p_1 = \gamma/\alpha,
  \end{equation}
  and it follows that \(q_1\) can take any values on \(\mathbb Z/d \mathbb Z\) and so \(\Psi(q_1)\) is maximally supported. However, it is also possible to reexpress the equation as:
  \begin{equation}
    p_1 + (\alpha/\beta) q_1 = \gamma/\beta.
  \end{equation}
It follows that \(p_1\) can also take any values on \(\mathbb Z/d \mathbb Z\)---\(\Psi(p_1)\) is also maximally supported. Therefore, one can always choose either a \(p_1\)- or \(q_1\)-basis such that the state is maximally supported and so is representable by a Gaussian function.

  We now consider adding another qudit such that the state becomes \({\Psi}_x(p_1,p_2,q_1,q_2)\). There are now two equations specified by \(\bs \Phi \cdot \bs x = \bs r\) and it follows that it is always possible to combine the two equations such that \(p_1\) and \(q_1\) are only in one equation and written in terms of each other (and generally the second degree of freedom):
  \begin{equation}
    \alpha q_1 + \beta p_1 + \gamma q_2 + \delta p_2 = \epsilon,
  \end{equation}
  for \(\alpha\), \(\beta\), \(\gamma\), \(\delta\) and \(\epsilon \in \mathbb{Z}/d\mathbb{Z}\).
  It will turn out that the \(\gamma q_2 + \delta p_2\) term is irrelevant. We can rewrite the above equation as:
  \begin{equation}
    q_1 + (\beta/\alpha) p_1 + (\gamma/\alpha) q_2 + (\delta/\alpha) p_2 = \epsilon/\alpha,
  \end{equation}
  if \(\alpha \ne 0\). Since there is no other equation specifying \(p_1\), this is an equation for a line on \(\mathbb Z/ d \mathbb Z\) and so \(\Psi\) is is maximally supported on \(q_1\). Otherwise, rewriting the above equation as:
  \begin{equation}
    p_1 + (\alpha/\beta) q_1 +(\gamma/\beta) q_2 + (\delta/\beta) p_2 = \epsilon/\beta,
  \end{equation}
  if \(\beta \ne 0\) shows that \(\Psi\) is maximally supported on \(p_1\). If \(\alpha = \beta = 0\) then both \(p_1\) and \(q_1\) are undetermined and so either representation produces a maximally supported state.

  The same procedure can be performed to find if \(q_2\) or \(p_2\) produce a maximally supported state. As can be seen, we are really just repeating the same procedure as we did when there was only one qudit because the other degrees of freedom have no impact on this determination. Expressing \(\Psi\) in the basis that is maximally supported in every degree of freedom means that it is therefore a Gaussian.
  
  Therefore, it follows that every degree of freedom (corresponding to a qudit) is maximally supported in either the \(p\)- or \(q\)- basis and so Eq.~\ref{eq:stabstatemixedrep} always describes stabilizer states for odd prime \(d\).\qed
\end{proof}

The form of Eq.~\ref{eq:stabstatemixedrep} is more general than Eq.~\ref{eq:stabstateposrep} because it does not depend on the support of the state. As we saw in the proof, this representation is generally not unique; for every qudit \(i\) that is not a position or momentum state, \(x_i\) can be either \(p_i\) or \(q_i\). However, if it is a position state then \(x_i = p_i\) and if it is a momentum state then \(x_i = q_i\); position and momentum states must be expressed in their conjugate representation in order to be captured by a Gaussian of the form in Eq.~\ref{eq:stabstatemixedrep} instead of Kronecker deltas.

The reason this mixed representation doesn't hold for non-prime odd \(d\) is that the coefficients above can be (multiples of) prime factors of \(d\) and so no longer produce ``lines'' in \(p_i\) or \(q_i\) that cover all of \(\mathbb Z/d \mathbb Z\). An alternative proof of this corollary that explores this case further is presented in the Appendix.

An example of the different classes of stabilizer states that are possible for odd prime \(d\), in terms of their support, is shown in Fig.~\ref{fig:quditstabstates_dis7}. There it can be seen that a stabilizer state is either maximally supported in \(p_i\) or \(q_i\), and is a Kronecker delta function in the other degree of freedom, or it is maximally supported in both. 

\begin{figure}[ht]
\includegraphics[scale=1.0]{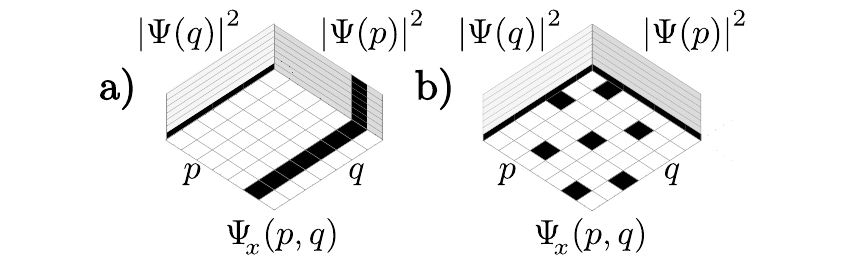}
\caption{The two classes of stabilizer states possible for odd prime \(d=7\) in terms of support: a) maximally supported in \(p\) or \(q\) and b) maximally supported in \(p\) and \(q\). The central grids denote the Wigner function \(\ketbra{\Psi}{\Psi}_x(p,q)\) of a stabilizer state \(\Psi\) with \(d=7\). The projection of this state onto \(p\)-space is shown in the upper right (\(|\Psi(p)|^2\)) and the projection onto \(q\)-space is shown in the upper left (\(|\Psi(q)|^2\)).}
\label{fig:quditstabstates_dis7}
\end{figure}

On the other hand, for odd non-prime \(d\), we see in Fig.~\ref{fig:quditstabstates_dis15} that another class is possible: stabilizer states that are maximally supported in neither \(p_i\) or \(q_i\). In fact, rotating the basis in any of the discrete angles afforded by the grid still does not produce a basis that is maximally supported (as discussed in the Appendix). Notice also, that Fig.~\ref{fig:quditstabstates_dis15}b shows that it is no longer true that a state that is maximally supported in only \(q_i\) or \(p_i\) is automatically a Kronecker delta when expressed in terms of the other.
\begin{figure}[ht]
\includegraphics[scale=1.0]{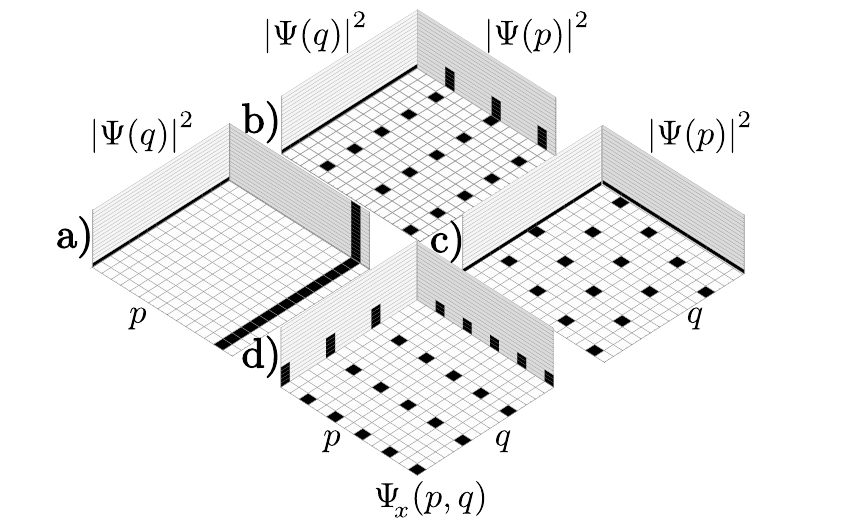}
\caption{The four classes of stabilizer states possible for odd non-prime \(d=15\) in terms of support: a) \& b) maximally supported in \(p\) or \(q\), c) maximally supported in \(p\) and \(q\), and d) not maximally supported in \(p\) or \(q\). The central grids denote the Wigner function \(\ketbra{\Psi}{\Psi}_x(p,q)\) for a stabilizer state \(\Psi\) with \(d=15\). The projection of this state onto \(p\)-space is shown in the upper right (\(|\Psi(p)|^2\)) and the projection onto \(q\)-space is shown in the upper left (\(|\Psi(q)|^2\)).}
\label{fig:quditstabstates_dis15}
\end{figure}

In summary, stabilizer states have Wigner function \(\delta_{\bs \Phi \cdot \bs x, \bs r}\) and, for odd prime \(d\), are Gaussians in mixed representation that lie on the Weyl phase space points \((\bs x_p, \bs x_q)\). Heuristically, they correspond to Gaussians in the continuous case that spread along their major axes infinitely. The only reason that they aren't always expressible as Gaussians in the mixed representation is that they sometimes ``skip'' over some of the discrete grid points due to the particular angle they lie along phase space for odd non-prime \(d\).

Wigner functions \(\Psi_x(\bs x)\) of stabilizer states propagate under \(\hat {V}\) as \(\Psi_x(\bs{\mathcal M}_{\hat {V}} (\bs x + \bs \alpha_{\hat {V}}/2) + \bs \alpha_{\hat {V}}/2)\), and this preserves the form of the state. In other words, Clifford gates take stabilizer states to other stabilizer states, as expected, just like in the continuous case Gaussians go to other Gaussians under harmonic evolution. It is also clear that stabilizer state propagation under Clifford gates can be expressed by a path integral at order \(\hbar^0\).

\section{Discrete Phase Space Representation of Universal Quantum Computing}
\label{sec:discreteuniversalcomp}

A similar statement to the one we made in Section~\ref{sec:contcenterchord}---that any operator can be expressed as an infinite sum of path integral contribution truncated at order \(\hbar^1\)---can be made in discrete systems. However, there is an important difference in the number of terms making up the sum. 

To see this we can follow reasoning that is similar to that employed in the continuous case. Namely, from Eq.~\ref{eq:discretesupofreflections} we see that any discrete operator can also be expressed as a linear combination of reflections, but unlike the continuous case, this sum has a finite number of terms. Since reflections can be expressed fully by the discrete path integral truncated at order \(\hbar^1\), as discussed previously, it follows that any unitary operator in discrete systems can be expressed as a finite sum of contributions from path integrals truncated at order \(\hbar^1\). Again, the same statement can be made by considering the chord representation in terms of translations.

Hence, quantum propagation in discrete systems can be fully treated by a \emph{finite} sum of contributions from a path integral approach truncated at order \(\hbar^1\).

To gather some understanding of this statement, we can consider what is necessary to add to our path integral formulation when we complete the Clifford gates with the T-gate, which produces a universal gate set.

The T-gate is generalized to odd \(d\)-dimensions by
\begin{equation}
  \hat {T} = \sum_{j \in \mathbb{Z}/d\mathbb{Z}} \omega^{\frac{(j-1)j}{4}} \ketbra{j}{j}.
\end{equation}
This gate can no longer be characterized by an \(\bs{\mathcal M}\) with integer entries. In particular,
\begin{equation}
\bs{\mathcal M}_{\hat {T}}  = \left( \begin{array}{cc} 1 & \frac{1}{2} \\ 0 & 1  \end{array} \right),
\end{equation}
and \(\bs \alpha_{\hat {T}} = \left( -\frac{1}{4}, 0 \right)\). This corresponds to
\begin{equation}
\bs{\mathcal B}_{\hat {T}} = \left( \begin{array}{cc} 0 & 0\\ 0 & -\frac{1}{4}\end{array} \right).
\end{equation}

Thus, the center function
\begin{equation}
T_x(x_p,x_q) = e^{-\frac{2 \pi i}{d} \frac{1}{4} (x_q-x_q^2)},
\end{equation}
corresponding to the phase shift Hamiltonian applied for only half the unit of time.

The operator can thus be written:
\begin{equation}
  \hat{T} = d^{-2} \sum_{\substack{x_p,x_q, \\ \xi_p,\xi_q \in \\ \mathbb{Z}/d\mathbb{Z}}} e^{ -\frac{2 \pi i}{d} \left[ \frac{1}{4} (x_q - x_q^2) - d \left(x_p \xi_q - x_q \xi_p \right) \right] } \hat {Z}^{\xi_p} \hat {X}^{\xi_q}.
\end{equation}

Though this operator is quadratic, it no longer takes the Weyl center points to themselves. This means that the \(\hbar^0\) limit of Eq.~\ref{eq:discretesemiproplowestorder} is now insufficient to capture all the dynamics because the overlap with any \(\ket{q}\) will now involve a linear superposition of partially overlapping propagated manifolds. It must therefore be described by a path integral formulation that is complete to order \(\hbar^1\). In particular,
\begin{equation}
  \label{eq:tgaterefl}
  \hat{T} = d^{-1} \sum_{\substack{x_p, x_q \in \\ (\mathbb{Z} / d \mathbb{Z})}} e^{-\frac{2 \pi i}{d} \frac{1}{4} (x_q-x_q^2)} \hat {R}( x_p, x_q ),
\end{equation}
where \(\hat {R}\) should be substituted by its path integral.

Note that this does not imply efficient classical simulation of quantum computation but quite the opposite. Indeed, for \(n\) qudits, there are \(d^{2n}\) terms in the sum above. While every Weyl phase space point has only a single associated path when acted on by Clifford gates, this is no longer true in any calculation of evolution under the T-gate. Eq.~\ref{eq:tgaterefl} expresses the T-gate as a sum over phase space operators (the reflections) evaluated on all the phase space points. Thus, it can be interpreted as associating an exponentially large number of paths to every phase space point, instead of the single paths found for Clifford gates. Therefore, any simulation of the T-gate naively necessitates adding up an exponential large sum over paths and so is comparably inefficient.

\section{Conclusion}
\label{sec:conc}

The treatment presented here formalizes the relationship between stabilizer states in the discrete case and Gaussians in the continuous case, which has often been pointed out~\cite{Gross06}. Namely, only Gaussians that lie along Weyl phase space points directly correspond to Gaussians in the continuous world in terms of preserving their form under a harmonic Hamiltonian, an evolution that is fully describable by truncating the path integral at order \(\hbar^0\). Furthermore, we showed that the Clifford group gates, generated by the Hadamard, phase shift and controlled-not gates, can be fully described by a truncation of their semiclassical propagator at lowest order. We found that this was because their Hamiltonians are harmonic and take Weyl phase space points to themselves. This proves the Gottesman-Knill theorem. The T-gate, needed to complete a universal set with the Hadamard, was shown not to satisfy these properties, and so requires a path integral treatment that is complete up to \(\hbar^1\). The latter treatment includes a sum of terms for which the number of terms scales exponentially with the number of qudits.

We note that our observations pertaining to classical propagation in continuous systems have long been very well known. In the continuous case, the Wigner function of a quantum state is non-negative if and only if the state is a Gaussian~\cite{Hudson74} and it has also long been known that quantum propagation from one Gaussian state to another only requires propagation up to order \(\hbar^0\)~\cite{Heller75}. Indeed, it has been shown that this is a continuous version of ``stabilizer state propagation'' in finite systems~\cite{Barnes04}, and is therefore, in principle, useful for quantum error correction and cluster state quantum computation~\cite{Nielsen06,Braunstein12}. It is also well known in the discrete case that quadratic Hamiltonians can act classically and be represented by symplectic transformations in the study of quantum cat maps~\cite{Hannay80,Rivas99,Rivas00} and linear transformations between propagated Wigner functions~\cite{Bianucci02}. Interestingly though, this latter work appears to have predated the discovery that stabilizer states have positive-definite Wigner functions~\cite{Gross06} and therefore, as far as we know, has not been directly related to stabilizer states and the \(\hbar^0\) limit of their path integral formulation, which is a relatively recent topic of particular interest to the quantum information community and those familiar with Gottesman-Knill. Otherwise, this claim has been pointed out in terms of concepts related to positivity and related concepts in past work~\cite{Cormick06,Mari12}.

We also note that our exploration of continous systems is not meant to explore the highly related topic of continuous-variable quantum information. Many topics therein apply to our discussion here, such as the continuous stabilizer state propagation we mentioned above. However, our intention in introducing the continuous infinite-dimensional case was not to address these topics but to instead relate the established continuous semiclassical formalism to the discrete case, and thereby bridge the notions of phase space and dynamics between the two worlds.

There is an interesting observation to be made of the weights of the reflections that make up the complete path integral formulation of a unitary operator. Namely, as is clear in Eqs.~\ref{eq:contsupofreflections} and~\ref{eq:discretesupofreflections}, the coefficients consist of the exponentiated center generating function multiplied by \(\frac{i}{\hbar}\). This is very similar to the form of the vVMG path integral in Eqs.~\ref{eq:contcenterrepvVMG} and~\ref{eq:discretecenterrepvVMG}. However, in Eqs.~\ref{eq:contsupofreflections} and~\ref{eq:discretesupofreflections}, reflections serve as the prefactors measuring the reflection spectral overlap of a propagated state with its evolute and the center generating actions provide the quantal phase. Thus, this formulation can be interpreted as an alternative path integral formulation of the vVMG, one consisting of reflections as the underlying classical trajectory only, instead of the more tailored trajectories that result from applying the method of steepest descents directly on an operator.

The fact that any unitary operator in the discrete case can be expressed as a sum consisting of a finite number of order \(\hbar^1\) path integral contributions, has the added interesting implication that uniformization---higher order \(\hbar\) corrections to the ``primitive'' semiclassical forms such as Eq.~\ref{eq:contcenterrepvVMG}---isn't really necessary in discrete systems. Uniformization is characterized by the proper treatment of coalescing saddle points and has long been a subject of interest in continuous systems where ``anharmonicity'' bedevils computationally efficient implementation. It seems that this problem isn't an issue in the discrete case since a fully complete sum with a finite number of terms, naively numbering \(d^2\) for one qudit, exists.

As a last point, there is perhaps an alternative way to interpret the results presented here, one in terms of ``resources''. Much like ``magic'' (or contextuality) and quantum discord can be framed as a resource necessary to perform quantum operations that have more power than classical ones, it is possible to frame the order in \(\hbar\) that is necessary in the underlying path integral describing an operation as a resource necessary for quantumness. In this vein, it can be said that Clifford gate operations on stabilizer states are operations that only require \(\hbar^0\) resources while supplemental gates that push the operator space into universal quantum computing require \(\hbar^1\) resources. The dividing line between these two regimes, the classical and quantum world, is discrete, unambiguous and well-defined.

\section{Acknowledgments}
The authors thank Prof. Alfredo Ozorio de Almeida for very fruitful discussions about the center-chord representation in discrete systems and Byron Drury for his help proof-reading and bringing ~\cite{Penney16} to our attention. This work was supported by AFOSR award no. FA9550-12-1-0046.

\section{Appendix}
\label{sec:appendix}

Gross proved that for odd prime \(d\)~\cite{Gross07}:
\begin{lemma}
Let \(\Psi\) be a state vector with positive Wigner function for odd prime \(d\). If \(\Psi\) is supported on two points, then it has maximal support.
\end{lemma}

With this lemma in mind, we can offer an alternative proof of Corollary~\ref{cor:stabstatemixedrep}:
\addtocounter{corollary}{-1}
\begin{corollary}
For odd prime \(d\), if \(\Psi\) is a stabilizer state then there always exists a mixed representation in position and momentum such that:
\begin{equation}
  \Psi_{\theta_{\beta \bs x},\eta_{\beta \bs x}}(\bs x) = \frac{1}{\sqrt{d}} \exp \left[\frac{2 \pi i}{d} \left( \bs x^T \bs \theta_{\beta \bs x} \bs x + \bs \eta_{\beta \bs x} \cdot \bs x \right) \right],
\end{equation}
where \(x_i\) can be either \(p_i\) or \(q_i\).
\end{corollary}
\begin{proof}

  We will show that for odd prime \(d\), every degree of freedom can only be fully supported or a Kronecker delta, for all other degrees of freedom fixed; WLOG we will consider a two-dimensional stabilizer state \(\Psi(q_1,q_2)\) and show that if \(\exists\, q'_1\) such that \(\Psi(q'_1, q_2) \ne 0\) \(\forall q_2\) then \(\Psi(q_1,q_2)\ne 0\) \(\forall\, q_1, q_2\) and vice-versa (if \(\exists\, q'_1\) s.t. \(\Psi(q'_1, q_2)\) is a delta function then \(\Psi(q_1,q_2)\) is a delta function in \(q_2\) \(\forall q_1\)). Therefore, if a degree of freedom is maximally supported in one degree of freedom for all others fixed, then it is maximally supported for all values of the other degrees of freedom. On the other hand, if it is a delta function in one degree of freedom for all others fixed, then it is a delta function for all values of the other degrees of freedom.

  Assume that for \(q_1, q_2 \in \{0, \ldots, d-1\}\), \(\exists\, q'_1, q''_1\) such that \(\Psi(q'_1,q_2) = 0\) for some \(q_2\) and \(\Psi(q''_1, q_2) \ne 0\) \(\forall q_2\). We proceed to prove by contradiction.

  Hence \(\Psi(q''_1, q_2) \equiv \Psi_{q''_1} \propto \left[ \frac{2 \pi i}{d} \left( \theta_{q''_1} q^2_2 + \eta_{q''_1} q_2 \right) \right]\) by Corollary~\ref{cor:stabstates} and \(\Psi(q'_1,q_2) \equiv \Psi_{q'_1}(q_2) \propto \delta_{q_2, q(q'_1)}\) for some \(q(q'_1)\in\mathbb{Z}/d\mathbb{Z}\) by~\cite{Gross06}.

  We can rotate in \(p_2\)-\(q_2\) space to form a new basis \(q^*_2\) in \(d\) discrete angles (since \(\theta_{q''_1} \in \mathbb{Z}/d\mathbb{Z}\)) such that \(\Psi_{q'_1}(q^*_2) \ne 0\) (since a delta function is not maximally supported only at the one angle perpendicular to it). Since there exists \((d-1)\) other values of \(q_1\) other than \(q'_1\), it follows that there exists at least one such angle such that \(\Psi_{q_1}(q^*_2) \ne 0\) \(\forall q_1, q^*_2 \in \{0,\ldots,d-1\}\). We define \(q^*_2\) as the basis that is rotated by this angle with respect to \(q_2\).

  By Corollary~\ref{cor:stabstates}, this means that
  \begin{equation}
    \Psi^*(q_1,q^*_2) \propto \exp( \theta'_{11} q^2_1 + \theta'_{22} {q^*_2}^2 + 2 \theta'_{12} q_1 q^*_2 + \eta'_1 q_1 + \eta'_2 q_2),
  \end{equation}
  where by \(\Psi^*\) we mean \(\Psi\) expressed in the new basis \(q^*_2\) in its second degree of freedom.
  Hence,
  \begin{equation}
    \label{eq:rotatedgaussianbasis}
    \Psi^*_{q_1}(q^*_2) \propto \exp \left[\frac{2 \pi i}{d} \left( \theta'_{q_1} {q^*_2}^2 + \eta'_{q_1} q^*_2 \right) \right] \exp \left[ \frac{2 \pi i}{d} \theta_{12} q_1 q^*_2 \right].
  \end{equation}

  Acting on this last equation to rotate back to \(q_2\), we must produce \(\Psi_{q''_1}(q_2) \propto \delta_{q_2,q(q''_1)}\). But then Eq.~\ref{eq:rotatedgaussianbasis} implies that \(\Psi_{q'_1}(q^*_2)\) must also be proportional to \(\delta_{q_2,q(q'_1)}\). This is a contradiction.

Therefore, if a degree of freedom is maximally supported for all others fixed, then it is maximally supported for all values of the other degrees of freedom and vice-versa. In the latter case, a position state in the \(i\)th degree of freedom can be represented as a Gaussian by using the \(p\)-basis where it becomes a plane wave (\(\theta_i = 0\)). In other words, one can always choose \(x_i\) to be \(p_i\) or \(q_i\) such that Eq.~\ref{eq:stabstatemixedrep} holds for odd prime \(d\).
  \qed
\end{proof}

Finally, the reason this result does not hold for odd non-prime \(d\) is that for discrete Wigner space there are \(d + 1\) unique angles minus all the prime factors of \(d\). For non-prime \(d\), there is more than one such prime factor and so there are cases when one cannot ``rotate'' away all non-maximally supported states.

\bibliography{biblio}{}
\bibliographystyle{unsrt}

\end{document}